\documentclass[submission,copyright,creativecommons]{eptcs}
\providecommand{\event}{FROM 2019} 
\usepackage{breakurl}             
\usepackage{underscore}           

\usepackage{amsmath, xcolor}%
\usepackage{amsfonts}%
\usepackage{amssymb}%
\usepackage{graphicx}%
\usepackage{dsfont}
\usepackage{amsthm}

\usepackage{tikz}
\usetikzlibrary{arrows,positioning,decorations.pathreplacing,calc} 
\tikzset{
    >=stealth',
    punkt/.style={
           rectangle,
           rounded corners,
           draw=black, very thick,
           text width=6.5em,
           minimum height=2em,
           text centered},
    empty/.style={
           rectangle,
           rounded corners,
           draw=white, very thick,
           minimum height=0em},
}

\newcommand{\Alice}{\textit{Alice}}
\newcommand{\Bob}{\textit{Bob}}
\newcommand{\Seller}{\textit{Seller}}
\newcommand{\Quote}{\textit{quote}}

\newcommand{\Title}{\textit{title}}
\newcommand{\Date}{\textit{date}}
\newcommand{\ok}{\textit{ok}}
\newcommand{\quit}{\textit{quit}}
\newcommand{\home}{\textit{homeAddress}}
\newcommand{\office}{\textit{officeAddress}}
\newcommand{\ISBN}{\textit{ISBN}}
\newcommand{\as}{\textit{as}}
\newcommand{\ab}{\textit{ab}}
\newcommand{\bs}{\textit{bs}}
\newcommand{\only}{\textit{only}}
\newcommand{\true}{\textit{true}}
\newcommand{\false}{\textit{false}}
\newcommand{\bool}{\textit{bool}}
\newcommand{\sid}{\textit{sid}}
\newcommand{\pid}{\textit{pid}}
\newcommand{\dom}{\textit{dom}}
\newcommand{\fn}{\textit{fn}}
\newcommand{\fv}{\textit{fv}}

\newcommand{\If}{\textsf{if}}
\newcommand{\Then}{\textsf{then}}
\newcommand{\Else}{\textsf{else}}
\newcommand{\End}{\textsf{end}}
\newcommand{\Int}{\textit{int}}
\newcommand{\String}{\textit{string}}

\newcommand{\Nat}{\textit{nat}}

\newtheorem{example}{Example}
\newtheorem{remark}{Remark}
\newtheorem{definition}{Definition}
\newtheorem{theorem}{Theorem}
\newtheorem{lemma}{Lemma}

\title{Probabilities in Session Types}
\author{Bogdan Aman \qquad\qquad Gabriel Ciobanu
\institute{Romanian Academy, Institute of Computer Science, Ia\c si, Romania}
\institute{Alexandru Ioan Cuza University of Ia\c si, Faculty of Computer Science, Romania}
\email{bogdan.aman@iit.academiaromana-is.ro \qquad gabriel@info.uaic.ro}
}
\def\titlerunning{Probabilities in Session Types}
\def\authorrunning{B. Aman \& G. Ciobanu}
\begin{document}
\maketitle

\begin{abstract}
This paper deals with the probabilistic behaviours of distributed systems 
described by a process calculus considering both probabilistic internal 
choices and nondeterministic external choices.
For this calculus we define and study a typing system which extends the 
multiparty session types in order to deal also with probabilistic behaviours.
The calculus and its typing system are motivated and illustrated by a 
running example.
\end{abstract}

\section{Introduction}

Probabilities allow uncertainty to be described in quantitative terms. 
If there are no uncertainties about how a system behaves, then its expected
behaviour has a 100\% chance of occurring, while any other behaviour would 
have no chance (i.e.,~0\% chance). Regarding the possible behaviours of 
a system, people working in artificial intelligence have used 
probability distributions over a set of events \cite{Halpern03}. In such 
an approach, the probabilities assigned to behaviours are real numbers 
from $[0,1]$ rather than values in $\{0,1\}$. In~\cite{Cooper14}, the authors 
made explicitly the assumption that probabilities are distributed over a 
restricted set of events, each of them corresponding to an equivalence 
class of events. We adapt these ideas to the framework of multiparty 
session types, and introduce probabilities assigned to actions and label 
selections.

An important feature of a probabilistic model is given by the distinction 
between nondeterministic and probabilistic choices~\cite{Segala95}. The 
nondeterministic choices refer to the choices made by an external 
process, while probabilistic choices are choices made internally by the 
process (not under control of an external process). Intuitively, a 
probabilistic choice is given by a set of alternative transitions, where 
each transition has a certain probability of being selected; moreover, 
the sum of all these probabilities (for each choice) is~$1$. To clarify 
the difference between nondeterministic and probabilistic choices, we 
consider a variant of the {\it two-buyers-seller} protocol \cite{Honda16} 
depicted in Figure \ref{diagram}.
Two buyers (\Alice\ and \Bob) wish to buy an expensive book (out of 
several possible ones) from a \Seller\ by combining their money (in 
various amounts depending on the amount of cash \Alice\ is willing to 
pay). The communications between them can be described in several steps. 
Firstly, \Alice\ sends (out of several choices) a book \Title\ (a string) 
or an \ISBN\ (a number) to the \Seller. The fact that \Alice\ chooses which 
book she wants to buy by sending the book title or book \ISBN\ is an 
example of a probabilistic choice, because it is under her control and 
preference (this is why in Figure \ref{diagram} we added probabilities to 
the possible choices of \Alice\ regarding the book). Then, \Alice\ waits 
for an answer regarding the quote of the book. This is a nondeterministic 
choice, because the choice of the answer received by \Alice\ is out of 
her control. This is due to the fact that the \Seller\ may provide 
different quotes depending on the buying history of \Alice\ and existing 
discounts. Next, \Seller\ sends back a \Quote\ (an integer) to \Alice\ and 
\Bob. \Alice\ tells \Bob\ how much she can contribute (an integer). 
Depending on the contribution of \Alice, \Bob\ notifies \Seller\ whether 
it accepts the quote or not. If \Bob\ accepts, he sends his home or 
office address (a string), and awaits from the \Seller\ a delivery date 
when the requested book will be received.

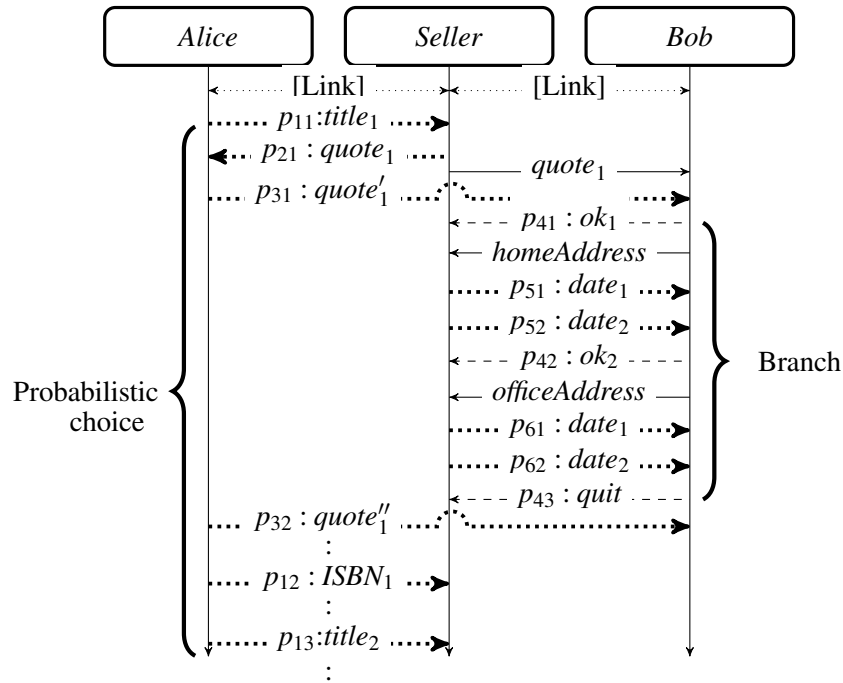
\begin{figure}[h]
\centering
\begin{tikzpicture}[scale=0.8,node distance=1cm, auto,]
 \node[punkt] (alice) at (0,0.5) {$\Alice$};
  \node[empty] (aliceDummy) at (0,-10.0) {};
 \node[punkt] (seller) at (4,0.5) {$\Seller$};
 \node[empty] (sellerDummy) at (4,-10.0) {};
  \node[punkt] (bob) at (8,0.5) {$\Bob$};
  \node[empty] (bobDummy) at (8,-10.0) {};
  \draw[->](alice.south)-- (aliceDummy.north);
  \draw[->](bob.south)-- (bobDummy.north);
  \draw[->](seller.south)-- (sellerDummy.north);
	
  	\draw[<->, dotted] (0.0,-0.4) -- node[midway,fill=white,yshift=-0.25cm] {[Link]} (4,-0.4);
  	\draw[<->, dotted] (4.0,-0.4) -- node[midway,fill=white,yshift=-0.25cm] {[Link]} (8,-0.4);
	
    \draw[->, line width=0.5mm, dotted] (0.0,-0.95) -- node[midway,fill=white,yshift=-0.25cm] {$p_{11}$:$\Title_1$} (4,-0.95);   

      \draw[<-, line width=0.5mm, dotted] (0.0,-1.5) -- node[midway,fill=white,yshift=-0.25cm] {$p_{21}:\Quote_1$} (4,-1.5); 
      \draw[->, solid] (4.0,-1.75) -- node[midway,fill=white,yshift=-0.25cm] {$\Quote_1$} (8,-1.75);
     \draw[->, line width=0.5mm, dotted] (0.0,-2.2) -- node[midway,fill=white,yshift=-0.25cm] {$p_{31}:\Quote'_1$} (3.8,-2.2) arc(180:0:0.25cm)
  -- (8,-2.2); 
        \draw[<-, dashed] (4.0,-2.6) -- node[midway,fill=white,yshift=-0.25cm] {$p_{41}:\ok_1$} (8,-2.6);
       \draw[<-, solid] (4.0,-3.1) -- node[midway,fill=white,yshift=-0.25cm] {$\home$} (8,-3.1);
       \draw[->, line width=0.5mm, dotted] (4.0,-3.75) -- node[midway,fill=white,yshift=-0.25cm] {$p_{51}:\Date_1$} (8,-3.75);
       \draw[->, line width=0.5mm, dotted] (4.0,-4.35) -- node[midway,fill=white,yshift=-0.25cm] {$p_{52}:\Date_2$} (8,-4.35);
  \begin{scope}[yshift=-2.15cm]
          \draw[<-, dashed] (4.0,-2.75) -- node[midway,fill=white,yshift=-0.25cm] {$p_{42}:\ok_2$} (8,-2.75);
       \draw[<-, solid] (4.0,-3.35) -- node[midway,fill=white,yshift=-0.25cm] {$\office$} (8,-3.35);
       \draw[->, line width=0.5mm, dotted] (4.0,-3.9) -- node[midway,fill=white,yshift=-0.25cm] {$p_{61}:\Date_1$} (8,-3.9);
       \draw[->, line width=0.5mm, dotted] (4.0,-4.5) -- node[midway,fill=white,yshift=-0.25cm] {$p_{62}:\Date_2$} (8,-4.5);
       \end{scope}
               \draw[<-, dashed] (4.0,-7.2) -- node[midway,fill=white,yshift=-0.25cm] {$p_{43}:\quit$} (8,-7.2);
  
  \draw[ultra thick, decoration={brace, raise=5pt,  amplitude=3mm}, decorate] (8,-2.6) -- (8,-7.2) node[midway,xshift=0.75cm] {Branch};
  
     \draw[->, line width=0.5mm, dotted] (0.0,-7.65) -- node[midway,fill=white,yshift=-0.25cm] {$p_{32}:\Quote''_1$} (3.8,-7.65) arc(180:0:0.25cm)
  -- (8,-7.65);   
  \node[empty] (skip1) at (2.0,-8.0) {$:$};

\begin{scope}[yshift=-7.2cm]
    \draw[->, line width=0.5mm, dotted] (0.0,-1.4) -- node[midway,fill=white,yshift=-0.25cm] {$p_{12}:\ISBN_1$} (4,-1.4);
      \node[empty] (skip2) at (2.0,-1.8) {$:$};
    \draw[->, line width=0.5mm, dotted] (0.0,-2.4) -- node[midway,fill=white,yshift=-0.25cm] {$p_{13}$:$\Title_2$} (4,-2.4); 
      \node[empty] (skip3) at (2.0,-2.9) {$:$};
\end{scope}    
    
\draw[ultra thick, decoration={brace, raise=5pt, mirror, amplitude=3mm}, decorate] (0.0,-1.0) -- (0.0,-9.8) node[midway,xshift=-2.75cm] {Probabilistic} node[midway,xshift=-2cm,yshift=-0.4cm] {choice};

\end{tikzpicture}
\vspace{-2ex}\caption{Dotted lines stand for probabilistic choices $p_{ij}$, dashed lines for branching, 
solid lines for deterministic choices, while double headed dotted lines for 
session initialization.}\label{diagram}

\vspace{-2ex}\end{figure}

One goal of the current research lines is to use a formal approach to 
describe in a rigorous way how distributed systems should behave, and then 
to design these systems properly in order to satisfy the behavioural constraints. 
In the last few years the focus has moved towards the quantitative study of the 
distributed systems behaviour to be able to solve problems that are not solvable 
by deterministic approaches (e.g., leader election problem~\cite{DengBook14}).

Probabilistic modelling is usually used to represent and quantify 
uncertainty in the study of distributed systems. Several probabilistic 
process calculi have been considered in the literature: probabilistic 
CCS~\cite{Hansson94}, probabilistic CSP~\cite{Lowe93}, probabilistic 
ACP~\cite{Andova99}, probabilistic asynchronous $\pi$-calculus~\cite{Herescu00}, 
PEPA~\cite{Hillston96}. The basic idea of these probabilistic process 
calculi is to include a probabilistic choice operator. Essentially, there 
are two possibilities of extending such an approach: either to replace 
nondeterministic choices by probabilistic choices, or to allow both 
probabilistic and nondeterministic~choices.

In this paper we consider the second alternative, and allow probabilistic 
choices made internally by the communicating processes (sending a value 
or a label), and also nondeterministic choices controlled by an external 
process (receiving a value or a label). Notice that in our 
operational semantics we impose that for each received value/label, the 
continuation of a nondeterministic choice is unique; thus, the 
corresponding execution turns out to be completely deterministic. 
We use a probabilistic extension of the process calculus presented in 
\cite{Honda16}, a calculus which is also an extension of the 
$\pi$-calculus~\cite{Milner99} for which the papers~\cite{Herescu00,Varacca07} 
present a probabilistic approach. For this calculus we define and study a 
typing system by extending the multiparty session types with both 
nondeterministic and probabilistic~behaviours.

Session types \cite{Honda93,Takeuchi94} and multiparty session 
types~\cite{Honda16} provide a typed foundation for the design of 
communication-based systems. The main intuition behind session types is 
that a communication-based application exhibits a structured sequence of 
interactions. Such a structure is abstracted as a type through an 
intuitive syntax which is used to validate programs. Session types are 
terms of a process algebra that also contains a selection construct (an 
internal choice among a set of branches), a branching construct (an 
external choice offered to the environment) and recursion. Session types 
are able to guarantee several properties in a session: (i) interactions 
never lead to a communication error (communication safety); (ii) channels 
are used linearly (linearity) and are deadlock-free (progress); (iii) the 
communication sequence follows a declared scenario (session fidelity, 
predictability).

While many communication patterns can be captured through such sessions, 
there are cases where basic multiparty session types are not able to 
capture interactions which involve internal probabilistic choices of the 
participants. Probabilities are used in the design and verification of 
complex systems in order to quantify unreliable or unpredictable behaviour, 
but also taken also into account when analyzing quantitative properties 
(measuring somehow the success level of the protocol). Overall, we study 
the nondeterministic and probabilistic choices in the framework of 
multiparty session types in order to understand better the quantitative 
aspects of uncertainty that might arise in communicating processes.

In the following, Section \ref{section:multipartyprocesses} presents the 
syntax and semantics of our probabilistic process calculus, and motivates 
the key ideas by using the two-buyers-seller protocol. 
Section~\ref{section:multipartysession} explains the global and local 
types, and the connection between them. Section \ref{section:typingsystem} 
describes the new typing system and presents the main results. 
Section~\ref{section:conclusion} concludes and discusses some related 
probabilistic approaches involving typing systems.

\section{Probabilistic Multiparty Session Processes}
\label{section:multipartyprocesses}

The most natural way to define a probabilistic extension of a process calculus 
consists of adding probabilistic information to some actions \cite{Glabbeek95}.
Probabilities are not attached to some actions, while others have 
probabilities (see \cite{Varacca07}).
When modelling the probabilistic behaviour of a distributed system, we 
should be able to model the fact that either the system or the environment 
chooses between several alternative behaviours. Moreover, when modelling 
such a system we should avoid to `approximate' the nondeterministic choice by a 
probabilistic distribution (very often a uniform distribution is used). 
For these reasons, we define a probabilistic extension of the process 
calculus used in \cite{Honda16} that combines both nondeterministic and 
probabilistic behaviours.
We actually define a calculus that puts together probabilistic internal 
choices (sending a value and selecting a label) with nondeterministic 
external choices (receiving a value and branching a process by using a 
selected value). In this setting, the nondeterministic actions of a process 
use information (values and labels) provided only by probability~actions.
The type system for this calculus is inspired from the synchronous 
multiparty session types \cite{Bejleri09}. As far as we know, our approach 
is new among the existing models used to formalize multiparty processes in 
the framework of multiparty session types.

\subsection{Syntax}

In what follows we use our variant of the two-buyers-seller protocol to 
illustrate some of the syntactic constructs defined afterwards.

\begin{example}\label{example1}

Let us note that the book to buy represents the choice of \Alice, and so 
she sends the title of a book (a string) or an \ISBN\ (a ten digit number). 
Since this is under her control and preference, it represents an example 
of a probabilistic choice.

$\Alice = 0.3:as!\langle ``\textit{War and Peace}''\rangle; \Alice_1$ 
$+$ $0.5:as!\langle ``\textit{The Art of War}''\rangle;\Alice_2$

\hspace{5ex}$+$ $0.2:as!\langle \textit{0195014766}\rangle;\Alice_3$. 

\noindent 
Here `$as$' denotes the channel used for the communication between \Alice\ 
and the \Seller. Actually, channels `$\as$' and `$\ab$' are used by \Alice\ 
to communicate with \Seller\ and with \Bob, while channel `$\bs$' is used by 
\Bob\ to communicate with \Seller.
We denote by \Alice$_i$ ($1\leq i \leq 3$) the different 
behaviours of \Alice\ after she sent her book choice. We use this index 
notation to keep track of the behaviours for each participant, to 
simplify the syntax and make it easier to read. The detailed description 
of all participants can be found in Example \ref{example2}.

When receiving the book orders, the \Seller\ expects the buyers 
sending him either a string representing a \Title\ of the book or a number 
representing an \ISBN. This behaviour is nondeterministic depending 
on the received information: \quad 
$Seller = \as?(\Title:\String); \Seller_1$ $+$ $\as?(\ISBN:\Nat);\Seller_2$. 
\end{example}

\noindent 
Informally, a session is a series of interactions between multiple parties 
serving as a unit of conversation. A session is established via a shared 
name representing a public interaction point, and consists of series of 
communication actions performed on fresh session channels. The syntax for 
processes is based on user-defined processes \cite{Honda16} extended with 
probabilistic choices. The syntax is presented in Table \ref{table:syntax}, 
where we use: probabilities $p_1$, $p_2$, $\ldots$; shared names $a$, $b$, 
$n$, $\ldots$ and session names $x$, $y$, $\ldots$; channels $s$, $t$, 
$\ldots$; expressions $e$, $e_i$, $\ldots$; labels $l$, $l_i$, $\ldots$, 
participants $q$, $\ldots$. We use symbols $q$ to name participants despite 
the fact that they are in reality numbers.

\begin{table}[h!]
\vspace{-1ex}  \centering
  \begin{tabular}{l@{\hspace{1ex}}lll@{\hspace{-2ex}}r}
   \hline
\vspace{1ex}
    {\it Processes} & $P$ & $::=$ & $\overline{a}[{\rm n}](\tilde{s}).P$   & (multicast session request)\\
    &  & $\shortmid$ & $a[q](\tilde{s}).P$   & (session acceptance)\\
    &  & $\shortmid$ & $\displaystyle \sum_{i\in I}p_i:s!\langle \tilde{e_i} \rangle; P_i$   & (value sending)\\
    &  & $\shortmid$ & $\displaystyle\sum_{i\in I}s? (\tilde{x_i}:\tilde{S_i}); P_i \qquad (\tilde{S_k}\neq \tilde{S_t}$, for $k,t \in I$, $k\neq t$)   & (value reception)\\
    &  & $\shortmid$ & $s!\langle\langle \tilde{s} \rangle\rangle; P$   & (session delegation)\\  
    &  & $\shortmid$ & $s?((\tilde{s})); P$   & (session reception)\\ 
    &  & $\shortmid$ & $\displaystyle\sum_{i\in I}p_i:s \lhd l_{i}; P_{i}$ $\qquad (l_k\neq l_t$, for $k,t \in I$, $k\neq t)$   & (label selection)\\ 
    &  & $\shortmid$ & $s \rhd \{l_{i}: P_{i}\}_{i \in I}$ $\qquad (l_k\neq l_t$, for $k,t \in I$, $k\neq t)$  & (label branching)\\
    &  & $\shortmid$ & \If\ $e$ \Then\ $P$ \Else\ $Q$   & (conditional branch)\\  
    & & $\shortmid$ & $P \mid Q$   & (parallel)\\
    & & $\shortmid$ & ${\bf 0}$   & (inaction)\\
    & & $\shortmid$ & $(\nu n)P$   & (hiding)\\   
    & & $\shortmid$ & $\mu X.P$   & (recursion)\\
    & & $\shortmid$ & $X$   & (variable)\\
    {\it Expressions} & $e$ & $::=$ & $v \shortmid e$ and $e' \shortmid$ not $e$ $\shortmid \ldots$ \\  
    {\it Values} & $v$ & $::=$ & $a \shortmid true \shortmid false$ $\shortmid \ldots$ \\  

    {\it Sorts} & $S$ & $::=$ & $\textit{bool} \mid \textit{nat} \mid \ldots$ & (value types)\\
    \hline 
   \end{tabular}
\vspace{-1ex}\caption{Syntax}\label{table:syntax}
\vspace{-2ex}\end{table} 

Excepting the primitives for value sending, value receiving and label 
selection, all the other constructs are from~\cite{Honda16}. The process 
$\overline{a}[{\rm n}](\tilde{s}).P$ sends along channel $a$ a request to 
start a new session using the channels~$\tilde{s}$ with participants $1$ 
$\ldots$ ${\rm n}$, where it participates as $1$ and continues as $P$. 
Its dual $a[q](\tilde{s}).P$ engages in a new session as participant~$q$. 
The communications taking place inside an established session are performed 
using the next six primitives: sending/receiving a value, session 
delegation/reception, and selection/branching. By using the 
delegation/reception pair, a process delegates to another one the 
capability to participate in a session by passing the channels associated 
with the session. The conditional branching establishes the continuation 
of an evolution based on the truth value of an expression~$e$. It is worth 
mentioning that the internal choices (sending a value and selecting a 
label) are probabilistically chosen, while the receiving values represent 
a nondeterministic choice (as external choice). The conditional branch, 
parallel and inaction are standard. A sequence of parallel composition is 
written~$\Pi_i P_i$. The syntax $(\nu n)P$ makes the name $n$ local to 
$P$. Interaction which can be repeated unboundedly is realized by 
recursion; as in~\cite{Bocchi14}, we do not use arguments when defining 
recursion. We often omit writing ${\bf 0}$ at the end of processes (e.g., 
$s!\langle \tilde{e_i} \rangle; {\bf 0}$ is written as $s!\langle 
\tilde{e_i} \rangle$).

The notions of identifiers (bound and free), process variables (bound and 
free), channels, alpha equivalence $\equiv_{\alpha}$ and substitution are 
standard. The bound identifiers are $\tilde{s}$ in multicast session 
request, session acceptance and session reception, $\tilde{x}_j$ in value 
reception and $n$ in hiding, while the bound process variable is $X$ in 
recursion. $\fv(P)$ and $\fn(P)$ denote the sets of free process variables 
and free identifiers of $P$, respectively.

\subsection{Operational Semantics}
Structural equivalence for processes is the least equivalence relation 
satisfying the following equations: 

\begin{table}[ht]
\vspace{-2ex} \centering
 \begin{tabular}{c}

\vspace{0.5ex}
$P \mid {\bf 0} \equiv P$ \qquad $P \mid Q \equiv Q \mid P$ \qquad $(P \mid Q)\mid R \equiv P \mid (Q \mid R)$\\

$(\nu n)P \mid Q \equiv (\nu n)(P \mid Q)$ if $n \notin \fn(Q)$\quad

$(\nu n)(\nu n') P \equiv (\nu n')(\nu n)P$\\

$(\nu n) {\bf 0} \equiv {\bf 0}$ \qquad $\mu X.{\bf 0}\equiv {\bf 0}$\qquad

$p_i: P + p_j : Q \equiv  p_j: Q + p_i : P$ \qquad $P + Q \equiv  Q + P$.\\

\end{tabular}
\vspace{-2ex}
\end{table}

We define the operational semantics in such a way that it distinguishes 
between probabilistic choices made internally by a process and 
nondeterministic choices made externally. This distinction allows us to 
reason about the evolutions of the system in which the nondeterministic 
actions of a process use only the data send by the probability actions. 
The operational semantics is given by a reduction relation denoted by 
$P\rightarrow_r Q$ (meaning $P$ reduces to $Q$ with probability~$r$) 
representing the smallest relation generated by the rules of Table 
\ref{table:semantics}\; (where $e \downarrow v$ means that expression $e$ is 
evaluated to value $v$).

\begin{table}[ht]
\vspace{-1ex}  
\centering
  \begin{tabular}{@{\hspace{0ex}}l@{\hspace{1ex}}c@{\hspace{0ex}}}
   \hline
\vspace{1ex}
$\overline{a}[{\rm n}](\tilde{s}).P_1 \mid \Pi_{q\in \{2..{\rm n}\}}a[q](\tilde{s}).P_q$ $\rightarrow_1 (\nu \tilde{s}) \Pi_{q\in \{1..{\rm n}\}}P_q$ & ({\sc Link})\\


$\displaystyle\sum_{i\in I}p_i: s!\langle \tilde{e_i} \rangle;P_i \mid \sum_{j\in J}s? (\tilde{x_j}:\tilde{S_j});P_j \rightarrow_{p_i} P_i \mid  P_j\{\tilde{v_i}/\tilde{x_j}\}$ ($\tilde{e_i} \downarrow \tilde{v_i},~ \tilde{v_j}:\tilde{S_j}$) & ({\sc Com}) \\

 


$s!\langle\langle \tilde{s} \rangle\rangle; P \mid s?((\tilde{s}));Q \rightarrow_{1} P \mid  Q$  & ({\sc Deleg}) \\


$\displaystyle\sum_{i\in I}p_i:s\lhd l_{i};P_{i} \mid s\rhd \{l_{j}:P_{j}\}_{j\in J}  \rightarrow_{p_i} P_i \mid P_j $ ($j \in J$)
 & ({\sc Label}) \\
 
 

\If\ $e$ \Then\ $P$ \Else\ $Q$ $\rightarrow_1$ $P$ ($e \downarrow true$) & ({\sc IfT})\\
\If\ $e$ \Then\ $P$ \Else\ $Q$ $\rightarrow_1$ $Q$ ($e \downarrow false$) & ({\sc IfF})\\
$\mu X.P \rightarrow_1 P\{\mu X.P/X\}$ & ({\sc Call})\\
$P \rightarrow_p P'$ implies $(\nu n)P \rightarrow_p (\nu n)P'$ & ({\sc Scope})\\
$P \rightarrow_p P'$ and $Q \not\rightarrow$ implies $P\mid Q \rightarrow_p P' \mid Q$ & ({\sc Par1})\\
$P \rightarrow_{p} P'$ and $Q \rightarrow_{q}Q' $ implies $ P\mid Q \rightarrow_{p \cdot q} P' \mid Q'$ & ({\sc Par2})\\
$P\equiv P'$ and $P' \rightarrow_p Q'$ and $Q' \equiv Q' $ implies $ P \rightarrow_p Q$ & ({\sc Struct})\\

\hline 
\end{tabular}
\vspace{-1ex}\caption{Operational Semantics}
\label{table:semantics}
\vspace{-2ex}\end{table}

Rule ({\sc Link}) describes a session initiation among ${\rm n}$ parties, 
generating $|\tilde{s}|$ fresh multiparty session channels. For simplicity, 
we consider that this rule has probability $1$; a normalization based on 
the possible reachable processes in one step is eventually needed (as done 
in \cite{Glabbeek95}). Rules ({\sc Com}), ({\sc Deleg}) and ({\sc Label}) 
are used to communicate values, session channels and labels. The values to 
be communicated in the rules ({\sc Com}) and ({\sc Label}) are chosen 
probabilistically, a fact illustrated by adding the probability of 
the consumed action to the transition of the reduction relation.
In both rules ({\sc Com}) and ({\sc Label}), the choice of the continuation 
process to be executed after sending or selecting is probabilistic, while 
when receiving or branching is nondeterministic. Inspired by 
\cite{Aldini00}, we add the conditions $\tilde{S_k}\neq \tilde{S_t}$ 
(meaning that the types of $\tilde{x_k}$ and $\tilde{x_t}$ are different ) 
and $l_k\neq l_t$ (meaning that $l_k$ and $l_t$ are different) in the rules 
({\sc Com}) and ({\sc Label}) to indicate that each value 
and label leads to a unique continuation. This means that a process of the 
form $s?(x:S);P + s?(x:S).Q$ is not possible in our syntax. Thus, the 
probability of the transition is equal with the probability of the 
sending/selecting process. 
It could be noticed from ({\sc Link}), ({\sc Com}) and ({\sc Label}) that 
the calculus is synchronous; this choice is made in order to simplify the 
presentation.

The rules ({\sc IfT}) and ({\sc IfF}) choose which branch to take
depending on the truth value of $e_i$. The rules ({\sc Scope}) and ({\sc 
Struct}) are standard. Rule ({\sc Par2}) is used to compose the evolutions 
of parallel processes, while rule ({\sc Par1}) is used to compose 
concurrent processes that are able to evolve with processes that are not 
able to evolve. In rule ({\sc Par1}), $Q \not\rightarrow$ means that the 
process $Q$ is not able to evolve by means of any rule (we say that $Q$ is 
a stuck process). Negative premises are used to denote the fact that 
passing to a new step is performed based on the absence of actions. The use 
of negative premises does not lead to an inconsistent set of rules. 
The following example illustrates how and when the rule ({\sc Par1}) is used.

\begin{example}[cont.]\label{example1bis}

Let us consider the process $P=\Alice \mid \Seller \mid \Bob$, where \Alice\ and \Seller\ have the definitions from Example \ref{example1}, while \Bob\ can have any form. By applying a ({\sc Com}) rule, we could 
have

\centerline{$\Alice \mid \Seller \rightarrow_{0.2} \Alice_3 \mid \Seller_2\{\textit{0195014766}/\ISBN\}$}

In order to illustrate the evolution of $P$, we need to add also \Bob\ to 
the above reduction. Notice that during this step \Bob\ is not able to 
interact neither with \Alice\ nor with \Seller. This is done by using the 
rule~({\sc Par1}), and so obtaining \quad
$P \rightarrow_{0.2} \Alice_3 \mid \Seller_2\{\textit{0195014766}/\ISBN\} \mid \Bob$ .

\end{example}

\begin{example}[cont.]\label{example2}
Let us consider an instance of the two-buyers-seller protocol in which
\Alice\ wants to buy one of the following two books:
\begin{itemize}
\item \textit{Title}: ``\textit{War and Peace}'' / \ISBN: $\textit{0140447938}$;
\item \textit{Title}: ``\textit{The Art of War}'' / \ISBN: $\textit{0195014766}$.
\end{itemize}
Firstly, \Alice\ sends to \Seller a book identifier (\Title\ or \ISBN), namely 
with probability~$0.3$ the book title ``\textit{War and Peace}'', with 
probability $0.5$ the book title ``\textit{The Art of War}'', and with 
probability $0.2$ the ISBN $\textit{0195014766}$ of the latter book. 
Then \Alice\ waits for \Seller\ to send a \Quote\ to both her and \Bob .
\Alice\ tells \Bob\ how much she can contribute (based on certain 
probabilities and the book she actually wants). For example, for the book 
``\textit{War and Peace}'' she is willing to participate with either 
$\Quote/2$ or $\Quote/3$ with the same probability $0.5$. We now describe 
formally the behaviour of \Alice\ as a process:

$Alice \stackrel{def}{=} \overline{a}[3] (\ab,\as,\bs).$

\hspace{6ex}$0.3:\as!\langle ``\textit{War and Peace}''\rangle;\as?(\Quote:\Nat);$

\hspace{8ex}$0.5:\ab!\langle \Quote/2 \rangle.P_1+0.5:\ab!\langle \Quote/3\rangle .P_1$

\hspace{4ex}$+$ $0.5:\as!\langle ``\textit{The Art of War}''\rangle;\as?(\Quote: \Nat);$

\hspace{8ex}$0.4:\ab!\langle \Quote/2\rangle.P_1+0.2:\ab!\langle \Quote/3\rangle.P_1+0.4:\ab!\langle \Quote/4\rangle.P_1$

\hspace{4ex}$+$ $0.2:\as!\langle \textit{0195014766}\rangle;as?(\Quote:\Nat);$

\hspace{8ex}$0.4:\ab!\langle \Quote/2\rangle.P_1+0.2:\ab!\langle \Quote/3\rangle.P_1+0.4:\ab!\langle \Quote/4\rangle.P_1$

\noindent Notice that the price options for the second book (searched 
either by \Title\ or \ISBN) are the same; however, this is just a coincidence
and not a requirement in our calculus. By using probabilities, it is 
possible to describe executions that may return different prices for the same 
title sold by the same \Seller, but possibly printed by different publishers.

Using this process (behaviour) of \Alice, we can find answers to questions like:

\begin{itemize}
\item What is the probability that \Alice\ buys ''\textit{The~Art~of~War}`` with $quote/3$?

This means that \Alice\ needs to execute

\centerline{$~~~~~~~~0.5:\as!\langle ``\textit{The Art of War}''\rangle;\as?(\Quote:\Nat);0.2:\ab!\langle \Quote/3\rangle.P_1$ }

\noindent with probability $0.5 \times 0.2 = 0.1$, or to execute

\centerline{$0.2:\as!\langle \textit{0195014766}\rangle;as?(\Quote:\Nat);0.2:\ab!\langle \Quote/3\rangle.P_1$}

\noindent with probability $0.2 \times 0.2 = 0.04$. Thus we get: 
\begin{itemize}
\item Answer: $0.5 \times 0.2+0.2 \times 0.2 = 0.14$.
\end{itemize}
\item What is the most probable choice made by \Alice ?
\begin{itemize}
\item Answer: ''The Art of War`` with $\Quote/2$ and $\Quote/4$, with probability $0.7 \times 0.4=0.28$.
\end{itemize}
\end{itemize}

\noindent \Alice\ is willing to contribute partially to the \Quote, 
contribution that is probabilistically chosen out of several possibilities, 
depending on the book \Alice\ intends to purchase. In process $P_1$, 
\Alice\ may perform the remaining transactions with \Seller\ and \Bob.

\end{example} 

\section{Global and Local Types}\label{section:multipartysession}

In what follows, the notion of probability already presented in the previous 
section scales up to the global types. 
Since the probabilities are static, the global types just need to check if 
the probabilities to execute certain actions are the desired ones. Usually 
session types lead to a unique description of a distributed system by means 
of processes. If we would simply incorporate probabilities in the session 
types as done for processes, this would be too restrictive as the slightest 
perturbation of the probabilities in the processes can make the system 
failing the prescribed behaviour. This is why in what follows we use 
probabilistic intervals in session types, allowing for several processes to 
be considered behavioural equivalent by having the same type.

\subsection{Global Types}\label{subsection:globaltypes}
The global types $G,G',\ldots$ presented in Table \ref{table:global} 
describe the global behaviour of a probabilistic multiparty session process. 
In what follows we use probabilistic intervals $\delta$ having one of 
the following forms $(c,d)$, $[c,d]$, $(c,d]$ or $[c,d]$, where $c,d \in 
[0,1]$ and $c\leq d$. For simplicity, we write $\delta = \lfloor c,d \rfloor$ with 
$\lfloor \in \{[,(\}$ and $\rfloor \in \{],)\}$. In what follows, we use 
also the addition of intervals defined as: if $\delta_1 = \lfloor c_1,d_1 
\rfloor$ and $\delta_2 = \lfloor c_2,d_2 \rfloor$ then $\delta_1+\delta_2 = 
\lfloor min(c_1+c_2,1),min(d_1+d_2,1) \rfloor$. 
If $\delta=[c,c]$, we use the shorthand notation $\delta=c$.

\begin{table}[ht]
\vspace{-2ex}  
\centering
\begin{tabular}{l@{\hspace{3ex}}lll@{\hspace{2ex}}r}
   \hline
   \vspace{1ex}
    {\it Global} & $G$ & $::=$ & $\displaystyle\sum_{i\in I}q \rightarrow_{\delta_i} q':k\langle S_i \rangle.G_i$   & (probValues)\\
    & &  $\shortmid$ & $q \rightarrow_{1} q':k\langle T@p \rangle.G'$   & (delegation)\\
   & & $\shortmid$ & $\displaystyle\sum_{i\in I}q \rightarrow_{\delta_i} q':k\{l_{i}:G_{i}\}$   & (probBranching)\\
    &  & $\shortmid$ & $G,G'$  & (parallel)\\
    & & $\shortmid$ & $\mu t.G$ & (recursive)\\
    & & $\shortmid$ & $t$ & (variable)\\
    & & $\shortmid$ & \End & (end)\\
    {\it Sorts} & $S$ & $::=$ & $\textit{bool} \mid \textit{nat} \mid \ldots$ & (value types)\\
    \hline
   \end{tabular}
\vspace{-1ex}\caption{Syntax of Global Types}\label{table:global}
\vspace{-2ex}\end{table}

Type $\displaystyle\sum_{i\in I}q \rightarrow_{\delta_i} q':k\langle S_i 
\rangle.G_i$ states that a participant $q$ sends with a probability in the 
interval $\delta_i$ a message of type $S_i$ to a participant $q'$ through 
the channel~$k$, and then the interactions described by $G_i$ take place. 
We assume that in each communication $q \rightarrow q'$ we have $q \neq q'$, 
i.e. we prohibit reflexive interactions. Type $q \rightarrow_{1} q':k\langle 
T@p \rangle.G'$ denotes the delegation of a session channel of type $T$ 
(called local type) with role $p$ (written as $T@p$). The local types are 
discussed in detail later.

Type $\displaystyle\sum_{i\in I}q \rightarrow_{\delta_i} 
q':k\{l_{i}:G_{i}\}$ says that participant $q$ sends with a probability in 
the interval~$\delta_i$ one label on channel $k$ to another participant 
$q'$. If $l_{i}$ is sent, evolution described by type $G_{i}$ takes 
place. Type $G,G'$ represents concurrent runs of processes specified by 
$G$ and~$G'$. Type $\mu t.G$ is a recursive type, where type variable 
$t$ is guarded in the standard way (they only appear under some prefix).  
Similar to the approach presented in \cite{Bocchi14}, we overload the 
notation $\mu$ as it is easy to see from the context if it precedes a 
process or a type. Type \End\ represents the termination of a process; 
we identify both $G$,\End\ and \End,$G$ with $G$.

In a probabilistic choice, identically behaved branches can be replaced 
by a single branch with a behaviour having the sum of the probabilities 
of the individual branches.

\begin{remark}\label{remark:simplify}
If all the possible interactions communicate the same types (all $S_i$ 
are identical), all select the same branch (all $l_i$ are identical), and the 
continuations after communications respect the same global type (all 
$G_i$ are identical), then the global systems 
can be simplified by using the following rules: 
\begin{itemize}

\item $\displaystyle\sum_{i\in I}q \rightarrow_{\delta_i} q':k\langle S_i 
\rangle.G_i$ is the same as $q \rightarrow_{\sum_{i\in I}\delta_i} q':k\langle S \rangle.G$
whenever $S=S_i$ and $G=G_i$ for all $i$;

\item $\displaystyle\sum_{i\in I}q \rightarrow_{\delta_i} q':k\{l_{i}:G_{i}\}$ 
is the same as $q \rightarrow_{\sum_{i\in I}\delta_i} q':k\{l_{i}:G_i\}$
whenever all $l_i$ are equal. 
\end{itemize}
\end{remark}

\noindent This means that if $\sum_{i\in I}\delta_i = 1$, then global types 
may contain only probabilities equal to $1$, namely a form similar to the 
global types in multiparty session types from~\cite{Honda16}. Therefore, 
for the processes of this particular type, all the results presented in 
\cite{Honda16} hold. 

\begin{example}[cont.]\label{example:global1}
Using the previous remark, the following is a global type of the 
two-buyers-seller protocol of Example \ref{example1}:

\hspace{10ex}$\Alice\rightarrow_{\lfloor 0.7,0.9 \rfloor} \Seller : \as\langle \String\rangle. G1$ 
$+ \Alice\rightarrow_{\lfloor 0.15,0.25 \rfloor} \Seller :\as\langle \Nat\rangle. G1$, \ where

\noindent
$G_1 = \Seller \rightarrow_1 \Alice : \as\langle \Int \rangle. \Seller \rightarrow_1 \Bob: 
\bs\langle \Int \rangle.$ 

\hspace{3ex}$\Alice \rightarrow_1 \Bob : \ab \langle \Int \rangle$.
$\Bob\rightarrow_{\lfloor 0.18,0.22 \rfloor} \Seller: \bs\{\ok_1 : \Bob\rightarrow_1 \Seller: \bs \langle \String \rangle.$

\hspace{31ex}$\Seller \rightarrow_1 \Bob: \bs \langle \Date\rangle.\End\}$

\hspace{23ex}$+\Bob\rightarrow_{\lfloor 0.27,0.31 \rfloor} \Seller: \bs\{\ok_2 : \Bob\rightarrow_1 \Seller: \bs \langle \String \rangle.$

\hspace{31ex}$\Seller \rightarrow_1 \Bob: \bs \langle \Date\rangle.\End\}$

\hspace{23ex}$+\Bob\rightarrow_{\lfloor 0.45,0.52 \rfloor} \Seller: \bs\{\quit.\End\}.$

This global type for \Alice\ is due to the fact that even if she has 
different book titles that she wants to buy, the global type only records 
the type of the sent value (namely a \String). Also, the fact that she 
behaves in a similar manner after sending the \Title, the global type can 
be reduced to a simpler form (according to the above remark).
\end{example}

\begin{example}\label{example:global2}
Let us consider now that \Alice\ decided that, instead of the books 
``\textit{War and Peace}'' and ``\textit{The Art of War}'', she wants the 
books ``\textit{Peter~Pan}'' and ``\textit{Robinson~Crusoe}'', and she is 
willing to pay different amount from the quote. More exactly,

$\Alice \stackrel{def}{=} \overline{a}[3] (\ab,\as,\bs).$

\hspace{6ex}$0.15:as!\langle ``\textit{Peter~Pan}''\rangle;\as?(\Quote:\Nat);$
$1:\ab!\langle \Quote/3\rangle .P_1$

\hspace{3ex}$+$ $0.65:\as!\langle ``\textit{Robinson~Crusoe}''\rangle;\as?(\Quote:\Nat);$
$0.35:\ab!\langle \Quote/3\rangle.P_1+0.65:\ab!\langle \Quote/4\rangle.P_1$

\hspace{3ex}$+$ $0.2:\as!\langle \textit{1593080115} \rangle;\as?(\Quote:\Nat);$
$0.45:\ab!\langle \Quote/2\rangle.P_1+0.55:\ab!\langle  \Quote/4\rangle.P_1$ .

\noindent It is worth mentioning that the two-buyers-seller protocol in 
which \Alice\ is described by this definition is well-typed using the 
same global type (the one from Example \ref{example:global1}) as the 
initial protocol of Example \ref{example1}. Therefore, several different 
processes may have the same global type. 
\end{example}

\subsection{Local Types}
\label{subsection:localtypes}
Local types $T,T',\ldots$ presented in Table \ref{table:local} describe 
the local behaviour of processes, acting as a link between global types 
and processes. 

\begin{table}[h]
\vspace{-1ex}  \centering
 \begin{tabular}{l@{\hspace{3ex}}lll@{\hspace{2ex}}r}
   \hline
   \vspace{1ex}
    {\it Local} & $T$ & $::=$ & $\displaystyle\sum_{i\in I}\delta_i:k!\langle S_i \rangle.T_i$   & (send) \\
    &  & $\shortmid$ & $\displaystyle\sum_{i\in I}k?(S_i).T_i $   & (receive) \\
 &  & $\shortmid$ & $k!\langle T@q \rangle.T'$   & (sessionDelegation) \\
    &  & $\shortmid$ & $k?(T@q).T' $   & (sessionReceive) \\
    &  & $\shortmid$ & $k\oplus \{\delta_i:(l_{i}:T_{i})\}_{i\in I}$  & (selection) \\
    &  & $\shortmid$ & $k\& \{l_{i}:T_{i}\}_{i\in I}$  & (branching)\\
    & & $\shortmid$ & $\mu t.T$ & (recursive)\\
    & & $\shortmid$ & $t$ & (variable)\\
    & & $\shortmid$ & \End & (end)\\     
    {\it Sorts} & $S$ & $::=$ & $bool \mid nat \mid \ldots$ & (value types)\\
    \hline 
   \end{tabular}
\vspace{-1ex}\caption{Syntax of Local Types}\label{table:local}
\vspace{-2ex}\end{table}

Type $\displaystyle\sum_{i\in I}\delta_i:k!\langle S_i \rangle.T_i$ 
represents the behaviour of sending with probability in the 
interval~$\delta_i$ a value of type $S_i$, and then behaving as described by 
type~$T_i$. Similarly, $\displaystyle\sum_{i\in I}k?(S_i).T_i$ is for 
nondeterministic receiving, and then continuing as described by local type $T_i$. 
The type $k!\langle T@q \rangle.T'$ represents the behaviour of delegating a 
session of type $T@q$, while $k?(T@q).T' $ describes the behaviour of 
receiving a session of type $T@q$.
Type $k\oplus \{\delta_i:(l_{i}:T_{i})\}_{i\in I}$ describes a branching: it 
waits for~$|I|$ options, and behaves as type $T_{i}$ if the $i$-th label is 
selected with probability in the interval $\delta_i$. Type $k\& 
\{l_{i}:T_{i}\}_{i\in I}$ represents the behaviour which 
nondeterministically selects one of the tags (say $l_{i}$), and then behaves 
as $T_{i}$. The rest is the same as for the global types, demanding type 
variables to occur guarded by a prefix. For 
simplicity, as done in~\cite{Honda16}, the local types do not contain the 
parallel composition.

\begin{example}
The following is a local type for the process \Alice\ presented in 
Example \ref{example2}:

\centerline{$\lfloor 0.7,0.9 \rfloor: \as!\langle \String \rangle.\as?(\Int).1:\ab!\langle \Int \rangle.T_1$ $+$ $\lfloor 0.15,0.25 \rfloor: \as!\langle \Nat \rangle.\as?(\Int).1:\ab!\langle \Int \rangle.T_1$ ,}

\noindent 
where $T_1$ is the local type of process $P_1$ from the definition of \Alice.
\end{example}
We define the projection of a global type to a local type for each~participant.

\begin{definition}
The projection for a participant $q$ appearing in a global type~$G$, written $G \upharpoonright q$, is inductively given as:

\begin{itemize}
\item[$\bullet$] $(q_1 \rightarrow_{1} q_2:k\langle T@p \rangle.G')\upharpoonright q =\begin{cases} 
k!\langle T@p \rangle.(G'\upharpoonright q) &\mbox{if } q=q_1 \neq q_2 \\
k?(T@p).(G'\upharpoonright q) & \mbox{if } q=q_2\neq q_1 \\
G'\upharpoonright q & \mbox{if } q\neq q_1  \mbox{ and } q \neq q_2
\end{cases} $;

\item[$\bullet$] $(\displaystyle\sum_{i\in I}q_1 \rightarrow_{\delta_i} q_2:k\langle S_i \rangle.G_i)\upharpoonright q =\begin{cases} 
\displaystyle\sum_{i\in I}\delta_i:k!\langle S_i \rangle.(G_i\upharpoonright q) &\mbox{if } q=q_1 \neq q_2 \\
\displaystyle\sum_{i\in I}k?(S_i).(G_i\upharpoonright q) & \mbox{if } q=q_2\neq q_1 \\
G_1\upharpoonright q & \mbox{if } q\neq q_1 \mbox{ and } q \neq q_2\\
& \forall i, j \in J,\  G_{i}\upharpoonright q=G_{j}\upharpoonright q 
\end{cases} $;

\item[$\bullet$] $(\displaystyle\sum_{i\in I}q_1 \rightarrow_{\delta_i} q_2:k\{l_{i}:G_{i}\})\upharpoonright q 
= \begin{cases} 
k\oplus \{\delta_i:(l_{i}:G_{i}\upharpoonright q)\}_{i\in I} &\mbox{if } q=q_1 \neq q_2\\
k\& \{l_{i}:G_{i}\upharpoonright q\}_{i\in I} & \mbox{if } q=q_2 \neq q_1 \\
G_{1}\upharpoonright q & \mbox{if } q\neq q_1  \mbox{ and } q \neq q_2\\
& \forall i, j \in J,\  G_{i}\upharpoonright q=G_{j}\upharpoonright q 
\end{cases} $;

\item[$\bullet$] $(G_1,G_2) \upharpoonright q  
= \begin{cases}
G_i \upharpoonright q & \mbox{if } q \in G_i \mbox{ and } q \notin G_j, i\neq j \in \{1,2\}\\
\End & \mbox{if } q \notin G_1 \mbox{ and } q \notin G_2
\end{cases}$;

\item[$\bullet$] $(\mu t.G) \upharpoonright q = 
\begin{cases}
\mu t. (G \upharpoonright q) & \mbox{if } G \upharpoonright q \neq \End \mbox{ or } G \upharpoonright q \neq t\\
\End & \mbox{otherwise}
\end{cases} $;

\item[$\bullet$] $t \upharpoonright q = t$ $\quad \bullet$ $\End \upharpoonright q =\End$.
\end{itemize}

When none of the side conditions hold, the projection is undefined.
\end{definition}

\begin{remark}
Regarding the check of linear usage of channels, the verification is 
similar to the one performed in \cite{Honda16}, noting that the 
probabilistic and nondeterministic choices are treated similar to the 
branching in \cite{Honda16}. However, due to the use of synchronous 
communications, the sequence of interactions follows more strictly the one 
of the global behaviour description, resulting in a simpler linear property 
than in \cite{Honda16}.
It should be said that in the branching clause, the projections of those 
participants different from $q_1$ and $q_2$ should generate an identical 
local type (otherwise undefined). 
\end{remark}
Hereafter we assume that global types are well-formed, i.e. $G 
\upharpoonright q$ is defined for all~$q$ occurring in~$G$. 

\section{Probabilistic Multiparty Session Types}
\label{section:typingsystem}

We introduce a typing system with the purpose of typing efficiently the 
probabilistic behaviours of our processes. This typing system uses a map 
from shared names to either their sorts $(S,S', \ldots)$, or to a special 
sort $\langle G \rangle$ used to type sessions. Since a type is inferred 
for each participant, we use notation $T@q$ (called located type) to 
represent a local type $T$ assigned to a participant~$q$. 
Using these, we~define

\centerline{$\Gamma$ $::=$ $\emptyset \,\mid\, \Gamma, x: S \,\mid\, \Gamma, 
a:\langle G \rangle \,\mid\, \Gamma, X:\Delta$ \qquad\qquad 
$\Delta$ $::=$ $\emptyset \,\mid\, \Delta,\tilde{s}:\{T@q\}_{q\in I}$ .}
\smallskip

\noindent
A sorting $(\Gamma,\Gamma', \ldots)$ is a finite map from names to sorts, 
and from process variables to sequences of sorts and types. Typing 
$(\Delta,\Delta',\ldots)$ records linear usage of session channels by 
assigning a family of located types to a vector of session channels. 
$\pid(G)$ stands for the set of participants occurring in~$G$, while 
$\sid(G)$ stands for the number of session channels in~$G$. We~write $\tilde{s}:T@q$ for a singleton typing $\tilde{s}:\{T@q\}$. 
Given two typings $\Delta$ and~$\Delta'$, their disjoint union is denoted 
by $\Delta,\Delta'$ (by assuming that their domains contain disjoint sets 
of session channels).

\begin{table*}[h!]
\vspace{-2ex}  \centering
   \begin{tabular}{@{\hspace{0ex}}c@{\hspace{-2ex}}r@{\hspace{0ex}}}
   \hline
\vspace{-2.5ex}
$~$\\
   \vspace{1.5ex}
   $\Gamma,x:S\vdash x:S$   \qquad
   $\Gamma\vdash \true,\false:\bool$ 
   & ({\sc TName}), ({\sc TBool})\\
   \vspace{1ex}
   $\dfrac{\displaystyle \Delta~\End~\only}{\displaystyle \Gamma\vdash {\bf 0}\rhd\Delta}$ \qquad
   $\dfrac{\displaystyle \Gamma \vdash e_i:\bool}{\displaystyle \Gamma \vdash e_1 \mbox{ or } e_2:\bool}$ & ({\sc TEnd}), ({\sc TOr})\\
   \vspace{1ex}
   $\dfrac{\displaystyle \begin{array}{c}
   \Gamma \vdash a:\langle G\rangle \quad \Gamma\vdash P \rhd \Delta,\tilde{s}:(G \upharpoonright 1)@1
   \quad \{1, \ldots, {\rm n}\}=\pid(G) \quad |\tilde{s}|=\sid(G)
   \end{array} }{\displaystyle \Gamma \vdash \overline{a}[{\rm n}](\tilde{s}).P\rhd \Delta}$ & ({\sc TMCast})\\   
   \vspace{1ex}
   $\dfrac{\displaystyle \begin{array}{c}
   \Gamma \vdash a:\langle G\rangle \quad \Gamma\vdash P \rhd \Delta,\tilde{s}:(G \upharpoonright q)@q \quad q\in \pid(G) \quad q \neq 1 \quad |\tilde{s}|=\sid(G)\end{array}}{\displaystyle \Gamma \vdash a[q](\tilde{s}).P\rhd \Delta}$ & ({\sc TMAccept})\\    
   \vspace{1ex}
   $\dfrac{{\displaystyle \forall i.\Gamma \vdash \tilde{e_i}:\tilde{S_i} \quad \forall i.\Gamma\vdash P_i \rhd \Delta,\tilde{s}:T_i@q} \quad \sum_{i\in I} p_i=1 \quad p_i \in \delta_i}{\displaystyle \Gamma \vdash \sum_{i\in I}p_i:s[k]!\langle \tilde{e_i} \rangle; P_i\rhd \Delta,\tilde{s}:\sum_{i\in I}\delta_i:k!\langle \tilde{S_i} \rangle;T_i@q}$ & ({\sc TSend})\\
   \vspace{1ex}   
   $\dfrac{\displaystyle \forall i.\Gamma, \tilde{x_i}:\tilde{S_i} \vdash P_i \rhd  \Delta,\tilde{s}:T_i@q}{\displaystyle \Gamma \vdash \sum_{i\in I}s[k]?(\tilde{x_i}:\tilde{S_i}); P_i\rhd \Delta,\tilde{s}:\sum_{i \in I}k?(\tilde{S_i});T_i@q}$ & ({\sc TReceive})\\    
   \vspace{1ex}
   $\dfrac{\Gamma \vdash P \rhd \Delta,\tilde{s}:T@q}{\Gamma \vdash s[k]!\langle\langle \tilde{t} \rangle\rangle; P\rhd \Delta,\tilde{s}:k!\langle \tilde{T'@q'} \rangle;T@q,\tilde{t}:T'@q'}$ & ({\sc TSDeleg})\\
   \vspace{1ex}   
   $\dfrac{\Gamma \vdash P \rhd \Delta,\tilde{s}:T@q,\tilde{t}:T'@q'}{\Gamma \vdash s[k]?((\tilde{t})); P\rhd \Delta,\tilde{s}:k?(T'@q');T@q}$ & ({\sc TSReceive})\\     
   \vspace{1ex}
   $\dfrac{{\displaystyle \forall i.\Gamma\vdash P_{i} \rhd \Delta,\tilde{s}:T_{i}@q} \quad \sum_{i\in I} p_i=1 \quad p_i \in \delta_i}{\displaystyle \Gamma \vdash \sum_{i\in I}p_i:s[k] \lhd l_{i}; P_{i}\rhd \Delta,\tilde{s}:k\oplus \{\delta_i:(l_{i}:T_{i})\}_{i\in I}@q}$ & ({\sc TSelect})\\
     \vspace{1ex}
   $\dfrac{\displaystyle \forall j.\Gamma\vdash P_{j} \rhd \Delta,\tilde{s}:T_{j}@q }{\displaystyle \Gamma \vdash s[k] \rhd \{l_{j}; P_{j}\}_{j\in J}\rhd \Delta,\tilde{s}:k\& \{l_{j}:T_{j}\}_{j\in J}@q}$ & ({\sc TBranch})\\   
    \vspace{1ex}
   $\dfrac{\displaystyle \Gamma\vdash P \rhd \Delta \qquad \Gamma\vdash Q \rhd \Delta'}{\displaystyle \Gamma\vdash P \mid Q \rhd \Delta, \Delta'}$ \qquad $\dfrac{\displaystyle \Gamma \vdash e \rhd bool \qquad \Gamma\vdash P \rhd \Delta \qquad \Gamma\vdash Q \rhd \Delta}{\displaystyle \Gamma\vdash \If~e~\Then~P~\Else~Q\rhd\Delta} $& ({\sc TConc}), ({\sc TIf})\\    
     \vspace{1ex}
    $\dfrac{\displaystyle \Gamma, a:\langle G \rangle \vdash P \rhd \Delta}{\displaystyle \Gamma\vdash (\nu a) P \rhd\Delta}$ \qquad $\dfrac{\displaystyle \Gamma \vdash P \rhd \Delta,\tilde{s}:\{T_i@i\}_{i \in I}}{\displaystyle \Gamma\vdash (\nu \tilde{s}) P \rhd\Delta}$
& ({\sc TNRes}), ({\sc TCRes})\\ 
   \vspace{1ex}
   $\dfrac{\Delta'~\End~\only}{\displaystyle \Gamma, X:\Delta\vdash X \rhd \Delta, \Delta'}$ \qquad $\dfrac{\displaystyle \Gamma,X:\Delta \vdash P\rhd \Delta}{\displaystyle \Gamma\vdash \mu X.P\rhd\Delta}$& ({\sc TVar}), ({\sc TRec})\\ 
      
    \hline 
   \end{tabular}
\vspace{-1ex}\caption{Typing System}\label{table:typing}
\vspace{-3ex}\end{table*}

The type assignment system for processes is given in Table 
\ref{table:typing}. We use the judgement $\Gamma \vdash P \rhd \Delta$  
saying that ``under the environment $\Gamma$, process $P$ has typing 
$\Delta$''. The rules ({\sc TName}),({\sc TBool}) and ({\sc TOr}) are for 
typing names and expressions. The rules
({\sc TMcast}) and ({\sc TMacc}) are for typing the session 
request and session accept, respectively. The type for $\tilde{s}$ is the 
projection on participant $q$ of the declared global type~$G$ for~$a$ 
in~$\Gamma$. It could be noticed that in rule ({\sc TMcast}) the projection 
is made on the participant requesting the session, while in ({\sc TMacc}) 
the projection is made for each of the $({\rm n}-1)$ accepting participants. 
The local type $(G\upharpoonright q)@q$ means that the participant~$q$ has $G 
\upharpoonright q$ (namely the projection of $G$ onto $q$) as its local 
type. The condition $|\tilde{s}| = \sid(G)$ ensures that the number of session 
channels meets those in $G$.

The rules ({\sc TSend}) and ({\sc TReceive}) are for sending and receiving 
values, respectively. As these rules require probabilistic and 
nondeterministic choices, the rules should check all the possible choices 
with respect to $\Gamma$. Since one of the channels appearing in 
$\tilde{s}$ (say $k$) is used for communication, we record~$k$ by using the 
name $s[k]$ as part of the typed process. In both rules, $q$ in 
$\tilde{s}:T_i@q$ ensures that each $P_i$ represents (being inferred as) 
the behaviour for participant $q$, and its domain should be $\tilde{s}$. 
Then the relevant type prefixes $\displaystyle\sum_{i\in I}\delta_i:k!\langle 
\tilde{S_i} \rangle;T_i@q$ for the output and $\displaystyle\sum_{i \in 
I}k?(\tilde{S_i});T_i@q$ for the input are composed in the session 
environment (as conclusion). The rules ({\sc TSDeleg}) and ({\sc TSReceive}) 
are for delegation of a session and its dual. They are similar to the 
rules ({\sc TSend}) and ({\sc TReceive}), except that here a vector of 
session channels is communicated instead of values. The carried type~$T'$ 
is located, making sure that the receiver takes the role of a specific 
participant (here~$q'$) in the delegated multiparty session. It should be 
noticed that in rule ({\sc TSDeleg}) the type of $\tilde{t}:T'@q'$ does not 
appear in the type of $P$, while it appears in rule ({\sc TSReceive}) 
meaning that it uses the channels of $P$. The rules ({\sc TSelect}) and 
({\sc TBranch}) are for typing selection and branching, respectively. 
Similar to ({\sc TSend}) and ({\sc TReceive}), these rules employ 
probabilistic and nondeterministic choices, respectively. This means that 
the rules should check all the possible choices with respect to $\Gamma$.

The rule ({\sc TConc}) composes two processes if their local types are 
disjoint. The rules ({\sc TIf}), ({\sc TEnd}), ({\sc TRec}) and ({\sc 
TVar}) are standard. The rules ({\sc TNRes}) and ({\sc TCRes}) represent 
the restriction rules for shared names and channel names, respectively. In 
({\sc TEnd}), ``$\Delta \ \End\ \only$'' means that $\Delta$ contains only 
{\sf end} types. 
\smallskip

As processes interact, their dynamics is formalized as in \cite{Honda16} by 
a reduction relation $\Rightarrow$ on typing~$\Delta$:
\begin{itemize}
\item $\tilde{s}:\{\displaystyle\sum_{i\in I}\delta_i:k!\langle \tilde{S_i} \rangle;T_i@q_1, \sum_{j \in J}k?(\tilde{S_{j}});T_{j}@q_2\}$ $\Rightarrow_{\delta_{k_1}}$

\hspace{10ex}$\tilde{s}:\{T_{k_1}@q_1,T_{k_2}@q_2,\ldots\}$, for $k_1\in I$, $k_2\in J$ and $S_{k_1}=S_{k_2}$;

\item $\tilde{s}:\{k!\langle T'@q' \rangle;T@q,, k?(T'@q');T''@q''\} \Rightarrow_1 \tilde{s}:\{T@q, T''@q''\}$

\item $\tilde{s}:\{k\oplus \{\delta_i:(l_{i}:T_{i})\}_{i\in I}@q_1,k\& \{l_{j}:T_{j}\}_{j\in J}@q_2,\ldots\}$ $\Rightarrow_{\delta_{k_1}}$

\hspace{10ex}$\tilde{s}:\{T_{k_1}@q_1,T_{k_2}@q_2,\ldots\}$, for $k_1\in I$, $k_2\in J$ and $S_{k_1}=S_{k_2}$;

\item $\Delta,\Delta' \Rightarrow_{p} \Delta,\Delta''$ if $\Delta' \Rightarrow_{p} \Delta''$.
\end{itemize}

\noindent The first rule corresponds to sending/receiving a value of type 
$\tilde{S_j}$ by the participant $q$, while the second rule corresponds to 
session delegation. The third rule illustrates the choice and reception of a 
label $l_j$ by the participant $q$.
The last rule is used to compose typings when only a part of a typing changes.
\smallskip

We present two basic properties of our type system: substitution and 
weakening. The substitution plays a central role in proving type 
preservation, while weakening allows introducing new entries in a typing.

\begin{lemma} \label{lemma:substweak} \ { } 
\begin{itemize}
\item[(1)] (substitution) $\Gamma,\tilde{x} : S \vdash P \rhd \Delta$ and $\Gamma \vdash \tilde{v} : S$ imply $\Gamma \vdash P\{\tilde{v}/\tilde{x}\}  \rhd \Delta$. 

\item[(2)] (type weakening) 
Whenever $\Gamma \vdash P \rhd \Delta$ is derivable, 
then its weakening is also derivable, \\ 
namely $\Gamma \vdash P \rhd \Delta,\Delta'$ for disjoint $\Delta'$, where 
$\Delta'$ contains only {end}.
\end{itemize} 
\end{lemma}
\begin{proof}
The proof is rather standard, similar to that presented in \cite{Honda16}.
\end{proof}

We now prove that our probabilistic typing system is sound, namely its 
type-checking rules prove only terms that are valid with respect to both 
structural congruence and operational semantics. In what follows, by 
inverting a rule we describe how the (sub)processes of a well-typed 
process can be typed. This is a basic property that is used in some papers 
when reasoning by induction on the structure of processes (see 
\cite{Bocchi14} and \cite{Honda16}, for instance).

\begin{theorem}[type preservation under equivalence] \label{theorem:struct} 
$\Gamma \vdash P\rhd \Delta$ and $P \equiv P'$ \ imply \ $\Gamma \vdash P'\rhd \Delta$ .
\end{theorem}

\begin{proof}
The proof is by induction on $\equiv$ , showing (in both ways) that if
one side has a typing, then the other side has the same typing.

\begin{itemize}
\item Case $P \mid {\bf 0} \equiv P$ .

$\Rightarrow$ Assume $\Gamma \vdash P \mid {\bf 0}\rhd \Delta$. By 
inverting the rule {\sc (TConc)}, we obtain $\Gamma \vdash P\rhd 
\Delta_1$ and $\Gamma \vdash {\bf 0} \rhd \Delta_2$, where $\Delta_1, \Delta_2 =\Delta$. By inverting the 
rule {\sc (TEnd)}, $\Delta_2$ is only {end} and $\Delta_2$ is such 
that $\dom(\Delta_1) \cap \dom(\Delta_2)=\emptyset$. Then, by 
weakening, we get that $\Gamma \vdash P\rhd \Delta$, where $\Delta= \Delta_1,\Delta_2$.

$\Leftarrow$ Assume $\Gamma \vdash P\rhd \Delta$. By rule {\sc (TEnd)}, 
it holds that $\Gamma \vdash {\bf 0} \rhd \Delta'$, where $\Delta'$ is 
only {end} and $\dom(\Delta) \cap \dom(\Delta')=\emptyset$. By 
applying rule {\sc (TConc)}, we obtain $\Gamma \vdash P \mid {\bf 0}\rhd 
\Delta,\Delta'$, and for $\Delta'=\emptyset$ we obtain $\Gamma \vdash P 
\mid {\bf 0}\rhd \Delta$, as required.
\end{itemize}
The remaining cases are proved in a similar manner.
\end{proof}

\noindent
According to the following theorem, if a well-typed process takes a reduction 
step of any kind, the resulting process is also well-typed.

\begin{theorem}[type preservation under reduction]\label{theorem:reduction} \ 

\centerline{$\Gamma \vdash P\rhd \Delta$ and $P\rightarrow_{p_i} P'$ \ imply \ $\Gamma \vdash P'\rhd \Delta'$, where $\Delta=\Delta'$ or $\Delta \Rightarrow_{\delta_i} \Delta'$ with $p_i \in \delta_i$.}
\end{theorem}

\begin{proof}
By induction on the derivation of $P \rightarrow_{p_i} P'$. There is a case 
for each operational semantics rule, and for each operational semantics 
rule we consider each typing system rule generating $\Gamma \vdash P \rhd \Delta$.

\begin{itemize}
\item Case {\sc (Com)}: $\displaystyle\sum_{i\in I}p_i: s!\langle 
\tilde{e_i} \rangle;P_i \mid \sum_{j\in J}s? (\tilde{x_j});P_j 
\rightarrow_{p_i} P_i \mid P_j\{\tilde{v_i}/\tilde{x_j}\}$ .

By assumption, $\Gamma \vdash \displaystyle\sum_{i\in I}p_i: s!\langle 
\tilde{e_i} \rangle;P_i \mid \sum_{j\in J}s? (\tilde{x_j});P_j \rhd 
\Delta$. By inverting the rule {\sc (TConc)}, we get $\Gamma \vdash 
\displaystyle\sum_{i\in I}p_i: s!\langle \tilde{e_i} \rangle;P_i \rhd 
\Delta_1 $, $\Gamma \vdash \displaystyle\sum_{j\in J}s? 
(\tilde{x_j});P_j \rhd \Delta_2$ with $\Delta=\Delta_1, \Delta_2$. 
Since these can be inferred only from {\sc (TSend)} and {\sc (TReceive)}, 
we know that $\Delta_1=\Delta'_1,\tilde{s}:\displaystyle\sum_{i\in 
I}\delta_i:k!\langle \tilde{S_i} \rangle;T_i@q_1$ and 
$\Delta_2=\Delta'_2,\tilde{s}:\displaystyle\sum_{j\in J}k?\langle 
\tilde{S_j} \rangle;T_j@q_2$. By inverting the rules {\sc (TSend)} and 
{\sc (TReceive)}, we get that $\forall i.\Gamma \vdash 
\tilde{e_i}:\tilde{S_i}$, $\sum_{i\in I} p_i=1$, $p_i \in \delta_i$, $\forall i.\Gamma\vdash 
P_i \rhd \Delta'_1,\tilde{s}:T_i@q_1$ and $ \forall j.\Gamma, 
\tilde{x_j}:\tilde{S_j} \vdash P_j \rhd \Delta'_2,\tilde{s}:T_j@q_2$. 
Assuming that $e_i \downarrow v_i$ and knowing that $\forall i.\Gamma 
\vdash \tilde{e_i}:\tilde{S_i}$, it implies that $\forall i.\Gamma \vdash 
\tilde{v_i}:\tilde{S_i}$. From $\Gamma \vdash \tilde{v_i}:\tilde{S_i}$ 
and $\Gamma, \tilde{x_j}:\tilde{S_j} \vdash P_j \rhd 
\Delta'_2,\tilde{s}:T_j@q_2$, by applying the substitution part of 
Lemma~\ref{lemma:substweak}, we get that $\Gamma \vdash P_j\{v_i/x_j\} \rhd 
\Delta'_2,\tilde{s}:T_j@q_2$. By applying the rule {\sc (TConc)}, we get 
$\Gamma \vdash P_i \mid P_j\{v_i/x_j\} \rhd 
\Delta'_1,\tilde{s}:T_i@q_1,\Delta'_2,\tilde{s}:T_j@q_2$. Using the 
reduction on types, we get $\Delta \Rightarrow_{\delta_i} \Delta'$, where 
$\Delta'=\Delta'_1,\tilde{s}:T_i@q_1,\Delta'_2,\tilde{s}:T_j@q_2$ and $p_i \in \delta_i$.

\item Case {\sc (Deleg)}: $s!\langle\langle \tilde{s} \rangle\rangle; P 
\mid s?((\tilde{s}));Q \rightarrow_{1} P \mid Q$ .

By assumption, $\Gamma \vdash s!\langle\langle \tilde{s} \rangle\rangle; P 
\mid s?((\tilde{s}));Q \rhd \Delta$. By inverting the rule {\sc (TConc)}, 
we get that $\Gamma \vdash s!\langle\langle \tilde{s} \rangle\rangle; P 
\rhd \Delta_1 $, $\Gamma \vdash s?((\tilde{s}));Q \rhd \Delta_2$ with 
$\Delta=\Delta_1, \Delta_2$. Since these can be inferred only from {\sc 
(TSDeleg)} and {\sc (TSReceive)}, we know that 
$\Delta_1=\Delta'_1,\tilde{s}:k!\langle \tilde{T'@q'} 
\rangle;T@q,\tilde{t}:T'@q'$ and 
$\Delta_2=\Delta'_2,\tilde{s}:k?(T'@q');T''@q''$. By inverting the rules 
{\sc (TSDeleg)} and {\sc (TSReceive)}, we get that $\Gamma \vdash P \rhd 
\Delta'_1,\tilde{s}:T@q$ and $\Gamma \vdash Q \rhd 
\Delta'_2,\tilde{s}:T''@q'',\tilde{t}:T'@q'$. By applying the rule {\sc 
(TConc)}, we get $\Gamma \vdash P \mid Q \rhd 
\Delta'_1,\tilde{s}:T@q,\Delta'_2,\tilde{s}:T''@q',\tilde{t}:T'@q''$. By 
using the reduction on types, we get that $\Delta \Rightarrow_{1} \Delta'$, 
where $\Delta'=\Delta'_1,\tilde{s}:T@q,\Delta'_2,\tilde{s}:T''@q'',\tilde{t}:T'@q'$.
\end{itemize}
The remaining cases are proved in a similar manner.
\end{proof}

A corollary of the type preservation result is the probabilistic-error 
freedom. An error is reached when a process performs an action that 
violates the constraints prescribed by its type. To formulate this 
property of probabilistic-error freedom, we extend the syntax by 
including a process {\sf error}, while the reduction rules for processes 
are extended as below. This is done to accommodate the fact that the 
processes with value sending and label selection in which the sum of all 
probabilities is different from $1$ generate an {\sf error}.

\begin{table}[ht]
\vspace{-2ex}  \centering
  \begin{tabular}{l@{\hspace{2ex}}c}
   \hline
\vspace{-1ex}
$~$\\
\vspace{1ex} 
$\displaystyle\sum_{i\in I}p_i: s!\langle \tilde{e_i} \rangle;P_i \mid \sum_{j\in J}s? (\tilde{x_j}:\tilde{S_j});P_j  \rightarrow_1$ {\sf error}  \qquad (if $\displaystyle\sum_{i\in I} p_i \neq 1$) & ({\sc ECom}) \\
\vspace{1ex}
$\displaystyle\sum_{i\in I}p_i:s\lhd l_{i};P_{i} \mid s\rhd \{l_{j}:P_{j}\}_{j\in J}  \rightarrow_1$ {\sf error} \quad (if $\displaystyle\sum_{i\in I} p_i \neq 1$)  & ({\sc ELabel}) \\
\hline 
\end{tabular}
\vspace{-2ex}\caption{Extending Operational Semantics with Rules for {\sf error}}
\label{table:semanticserror}
\vspace{-2ex}\end{table}

\begin{theorem}[probabilistic-error freedom] If $\Gamma \vdash P\rhd 
\Delta$ and $P\rightarrow_{p_i} P'$, then $P' \neq {\sf error}$.
\end{theorem}

\begin{proof}
We assume that $P \neq {\sf error}$, and proceed by case analysis on the 
reduction~$P\rightarrow_{p_i} P'$.
If the last reduction is by one of the rules of Table \ref{table:semantics} 
then $P' \neq {\sf error}$ since these rules do not introduce {\sf error} 
processes. Also, by using Theorem \ref{theorem:reduction}, we are able to 
show that $\Gamma \vdash P' \rhd \Delta'$ for some $\Delta'$ (obtained 
by some reduction from $\Delta$).

The only reductions introducing error processes are provided by the rules 
of Table \ref{table:semanticserror}. We consider only one case (as the 
other is treated in a similar manner). Consider the rule ({\sc ECom}) 
applied to $P$ having the form $\displaystyle\sum_{i\in I}p_i: s!\langle 
\tilde{e_i} \rangle;P_i \mid \sum_{j\in J}s? (\tilde{x_j}:\tilde{S_j});P_j$. Then by 
({\sc ECom}) we have $\displaystyle\sum_{i\in I} p_i \neq 1$. By 
hypothesis, $P$ is well-typed. By using the typing rules ({\sc TSend}) and 
({\sc TReceive}) of Table~\ref{table:typing}, process~$P$ can be typed by 
using the condition $\displaystyle\sum_{i\in I} p_i = 1$ which contradicts 
the fact that rule ({\sc ECom}) can be applied. The fact that none of the 
reductions introducing errors can be applied means that the result holds.
\end{proof}

\vspace{-2ex}

By the correspondence between local types and global types given in 
Section \ref{subsection:localtypes}, these results guarantee that 
interactions between typed processes follow exactly the interactions 
specified in a global type.

\section{Conclusion}\label{section:conclusion}

We have defined and studied a typing system extending the (synchronous 
version of) multiparty session types to deal also with probabilistic and 
nondeterministic choices.  We proposed a process calculus considering both 
the probabilistic internal choices (sending a value and selecting a label) 
with the nondeterministic external choices (receiving a value and 
branching a process by using a selected value). We used a system inspired 
from the synchronous calculus presented in \cite{Bejleri09}, but avoiding 
the use (and typing) of queues presented in \cite{Bejleri09}. The calculus 
from \cite{Bejleri09} has been modified in \cite{Coppo16} and 
\cite{Scalas17} by using channels with roles, and so eliminating the need 
to use the notation $T@q$ for delegation. However, we feel that this 
notation for delegation makes the rules easier to read; thus, we keep it in 
our typing system.

The approach presented in this paper has attractive properties and features. 
It retains the classical approach (type system), and it is specified in 
such a way to satisfy the axioms of a standard probability theory for 
computing the probability of a behaviour. As far as we know, in the field 
of session types there is no other related work.

Several formal tools have been proposed for probabilistic reasoning.  
Some approaches concern the use of probabilistic logics. In 
\cite{Cooper14}, terms are assigned probabilistically to types via 
probabilistic type judgements, and from an intuitionistic typing system is 
derived a probabilistic logic as a subsytem~\cite{Warrell16}.

In \cite{Varacca07} there are proposed two semantics of a probabilistic 
variant of the $\pi$-calculus. For these, the types are used to identify 
a class of nondeterministic probabilistic behaviours which can preserve 
the compositionality of the parallel operator in the framework of event 
structures. The authors claim to perform an initial step towards a good 
typing discipline for probabilistic name passing by employing Segala 
automata \cite{Segala95} and probabilistic event structures. In comparison
with them, we simplify the approach and work directly with processes, 
giving a probabilistic typing in the context of multiparty session types.

\nocite{*}
\bibliographystyle{eptcs} 
\bibliography{FROM19ref} 
\end{document}